\documentclass[11pt]{article}
\usepackage{fullpage,epsf,amssymb,bbm,amsmath,exscale}

\input amssym.def
\input amssym

\newcommand{\E}{\mathbf{E}}
\renewcommand{\Pr}{\mathbf{Pr}}
\newcommand{\qed}{\hspace{1em}\hfill\rule{.5em}{.5em}}

\newtheorem{theorem}{Theorem}
\newtheorem{observation}{Observation}
\newtheorem{lemma}[theorem]{Lemma}

\newtheorem{remark}{Remark}
\newenvironment{proof}{\paragraph{Proof:}\
  }{\mbox{}\hfill\qed\medskip}

\begin{document}

%{\large\bf Retrieval and approximate membership in less space}
%\title{\Large Retrieval and Approximate Membership in Less Space}
\title{Succinct Data Structures for Retrieval\\ and Approximate Membership%
\thanks{The main ideas for this paper were conceived while the authors were
participating in the 2006 Seminar on Data Structures at IBFI Schloss Dagstuhl, Germany.}}

\author{Martin Dietzfelbinger\protect\thanks{Faculty of Computer Science and Automation, Technische Universit{\"a}t Ilmenau, P.O.Box 100565, 98684 Ilmenau, Germany, email:
{\tt martin.dietzfelbinger@tu-ilmenau.de}} \and Rasmus Pagh\thanks{Computational Logic and Algorithms Group, IT University of Copenhagen, Rued Langgaards Vej 7, 2300 K{\o}benhavn S,
Denmark, email: {\tt pagh@itu.dk}}}

\maketitle

\begin{abstract}
The {\em retrieval problem} is the problem of associating data with keys in a set. Formally, the data structure must store a function $f\colon U\to \{0,1\}^r$ that has specified values on the elements of a given set $S\subseteq U$, $|S|=n$, but may have any value on elements outside $S$. Minimal perfect hashing makes it possible to avoid storing the set $S$, but this induces a space overhead of $\Theta(n)$ bits in addition to the $nr$ bits needed for function values. In this paper we show how
% an approach of (\emph{Chazelle et al., SODA 2004}) can be extended 
to eliminate this overhead. Moreover, we show that for any $k$ query time $O(k)$ can be achieved using space that is within a factor $1+e^{-k}$ of optimal, asymptotically for large $n$. If we allow logarithmic evaluation time, the additive overhead can be reduced to $O(\log \log n)$ bits whp. The time to construct the data structure is $O(n)$, expected. 
% Can we show whp.? 
A main technical ingredient is to utilize existing tight bounds on the probability of almost square random matrices with rows of low weight to have full row rank.
In addition to direct constructions, we point out a close connection between retrieval structures and hash tables where keys are stored in an array and some kind of probing scheme is used.
Further, we propose a general reduction that transfers the results on retrieval into analogous results on \emph{approximate membership}, a problem traditionally addressed using Bloom filters. Again, we show how to eliminate the space overhead present in previously known methods, and get arbitrarily close to the lower bound. The evaluation procedures of our data structures are extremely simple (similar to a Bloom filter). For the results stated above we assume free access to fully random hash functions. However, we show how to justify this assumption using extra space 
$o(n)$ to simulate full randomness on a RAM.
\end{abstract}

%{\small Version: 20071119-submitted-bloomier.tex, 15:44 pm}

\newpage

%%%%%%%%%%%%%%%%%%%%%%%%%%%%%%%%%%%%%%%%%%%%%%%%%% 

\section{Introduction}

Suppose we want to build a data structure that is able to distinguish between girls' and boys' names, in a collection of $n$ names. Given a string not in the set of names, the data structure may return any answer. It is clear that in the worst case this data structure needs at least $n$ bits, even if it is given access to the list of names. The previously best solution (implicit in~\cite{p:MWHC96}) that does not require the set of names to be stored requires around $1.23n$ bits. Surprisingly, as we will see in this paper, $n+o(n)$ bits is enough, still allowing fast queries. If ``global'' hash functions, shared among all data structures, are available the space usage drops all the way to $n+O(\log \log n)$ bits whp. This is a rare example of a data structure with non-trivial functionality and a space usage that essentially matches the entropy lower bound.

\subsection{Problem definition}

The {\em dictionary problem\/} consists of storing a set $S$ of $n$ keys, and $r$ bits of data associated with each key. A {\em lookup\/} query for $x$ reports whether or not $x\in S$, and in the positive case reports the data associated with $x$. We will denote the size of $S$ by $n$, and assume that keys come from a set $U$ of size $n^{O(1)}$. In this paper, we restrict ourselves to the {\em static\/} problem, where $S$ and the associated data are fixed and do not change. We study two relaxations of the static dictionary problem that allow data structures using less space than a full-fledged dictionary:
\begin{itemize}\setlength{\itemsep}{0pt}
\item The {\em retrieval problem\/} differs from the dictionary problem in that the set $S$ does not need to be stored. A retrieval query on $x\in S$ is required to report the data associated with $x$, while a retrieval query on $x\not\in S$ may return any $r$-bit string.
\item The {\em approximate membership problem\/}  consists of storing a data structure that supports membership queries in the following manner: For a query on $x\in S$ it is reported that $x\in S$. For a query on $x\not\in S$ it is reported with probability at least $1-\varepsilon$ that $x\not\in S$, and with probability at most $\varepsilon$  that $x\in S$ (a ``false positive''). For simplicity we will assume that $\varepsilon$ is a negative power of~2.
\end{itemize}
The model of computation is a unit cost RAM with a standard instruction set. For simplicity we assume that a key fits in a single machine word, and that associated values are no larger than keys. Some results will assume free access to fully random hash functions, such that any function value can be computed in constant time. (This is explicitly stated in such cases.)

\subsection{Motivation}

The approximate membership problem has attracted significant interest in recent years due to a number of applications, mainly in distributed systems and database systems, where false positives can be tolerated and space usage is crucial (see~\cite{BM:02} for a survey). Often the false positive probability that can be tolerated is relatively large, say, in the range $1\%-10\%$, which entails that the space usage can be made much smaller than what would be required to store $S$ exactly.

The retrieval problem shows up in situations where the amount of data associated with each key is small, and
it is either known that queries will only be asked on keys in $S$, or where the answers returned for keys not in $S$ do not matter. As an example, suppose that we have ranked the URLs of the World Wide Web on a $2^r$ step scale, where $r$ is a small integer. Then a retrieval data structure would be able to provide the ranking of a given URL, without having to store the URL itself. The retrieval problem is also the key to obtaining a space-optimal RAM data structure that is able to answer range queries in constant time~\cite{range-1D,range1d2}.

\subsection{Previous results}

\paragraph{Approximate membership.}
The study of approximate membership was initiated by Bloom~\cite{bloom} who described the {\em Bloom filter\/} data structure which provides an elegant, near-optimal solution to the problem: 
The data structure is a bit array. Use $k=\log_2(1/\varepsilon)$ hash functions to associate each key with $k$ randomly located bits in the array. Set these bits to 1 for all $x\in S$, and put 0s elsewhere. On a query for $x$, the Bloom filter reports that $x\in S$ if and only if all bits associated with $x$ are 1. Bloom showed\footnote{Bloom used a certain simplifying assumption, independence of certain slightly correlated events, that has since been justified, see~\cite{BM:02}.} that a space usage of $n\log_2(1/\varepsilon)\log_2 e$ bits suffices for a false positive probability of~$\varepsilon$. Carter \emph{et al}.~\cite{MR80h:68037} showed that $n\log_2(1/\varepsilon)$ bits are required for solving the approximate membership problem when $|U|\gg n$ (see Appendix~\ref{app:approx-lower} for details).
Thus, the analysis of~\cite{bloom} shows that Bloom filters have space usage within a factor $\log_2 e\approx 1.44$ of the lower bound, which is tight.

Another approach to approximate membership is {\em perfect hashing}. A minimal perfect hash function for $S$ maps the keys of $S$ bijectively to $[n]=\{0,\ldots,n-1\}$. Hagerup and Tholey~\cite{HT01} showed how to store a minimal perfect hash function $h$ in a data structure of $n\log_2 e + o(n)$ bits such that it can be evaluated on a given input in constant time. This space usage is the best possible, up to the lower order term. Now store an array of $n$ entries where, for each $x\in S$, entry $h(x)$ contains a $\log_2(1/\varepsilon)$-bit hash signature $q(x)$. When looking up a key $x$, we answer $x\in S$ if and only if the hash signature at entry $h(x)$ is equal to $q(x)$. The origin of this idea is unknown to us, but it is described e.g.~in~\cite{BM:02}. The space usage for the resulting data structure differs from the lower bound $n\log_2(1/\varepsilon)$ by the space required for the minimum perfect hash function, and improves upon Bloom filters when $\varepsilon\leq 2^{-4}$ and $n$ is sufficiently large.

Mitzenmacher~\cite{Mitz:02} considered the {\em encoding\/} problem where the task is to represent and transmit an approximate set representation (no fast queries required). However, even in this case existing techniques have a space overhead similar to that of the perfect hashing approach. 

\paragraph{Retrieval.}
The retrieval problem has traditionally been addressed through the use of perfect hashing. Using the Hagerup-Tholey data structure yields a space usage of $nr + n\log_2 e + o(n)$ bits with constant query time.

Recently, Chazelle \emph{et al}.~\cite{CKRT04} presented a different approach to the problem based on an idea similar to that of a Bloom filter: Each key is associated with $k=O(1)$ locations in an array with $O(n)$ entries of $r$ bits. The answer to a retrieval query on $x$ is found by combining the values of entries associated with $x$, using bit-wise XOR. In place of the XOR operation, any abelian group operation may be used. In fact, this idea was used earlier by Majewski, Wormald, Havas, and Czech~\cite{p:MWHC96} and by Seiden and Hirschberg~\cite{HS94}  to address the special case of order-preserving minimal perfect hashing. It is not hard to see that these data structure in fact solve the retrieval problem. The main result of~\cite{p:MWHC96} is that for $k=3$ a space usage of around $1.23nr$ bits is possible, and this is the best possible using the construction algorithm of~\cite{CKRT04,p:MWHC96} (other values of $k$ give worse results).

The approach of these papers does not give a data structure that is more efficient than perfect hashing, asymptotically for large~$n$, but  the simplicity  and the lack of lower order terms in the space usage that may dominate for small $n$  makes it interesting from a practical viewpoint.  A particular feature is that (like for Bloom filters) all memory lookups are nonadaptive, i.e., the memory addresses can be determined from the query only. This can be exploited by modern CPU architectures that are able to parallelize memory lookups (see e.g.~\cite{ZuHeBo:DAMON:06}). 

In fact, Chazelle \emph{et al}.~also show how approximate membership can be incorporated into their data structure by extending array entries to $r+\log_2(1/\varepsilon)$ bits. This generalized data structure is called a {\em Bloomier filter}. Again, the space usage is a constant factor higher, asymptotically, than the solution based on perfect hashing.

\subsection{New contributions}

Our first contribution shows that the approach of~\cite{CKRT04,p:MWHC96,HS94} can be used to achieve space for retrieval that is very close to the lower bound:

\begin{theorem}\label{thm:main}
For any $\gamma>0$, $r=O(\log n)$, and any sufficiently large $n$ there exist data structures for the retrieval problem having the following space and time complexity on a unit cost RAM with free access 
to a fully random hash function\emph{:}
\begin{description}\setlength{\itemsep}{0pt}
\item[{\rm (a)}] Space $nr+O(\log\log n)$ bits whp.%
\footnote{``whp.'' means with probability $1-O(\frac{1}{\mathrm{poly}(n)})$.}%
, query time $O(\log n)$, expected construction time $O(n^3)$.
\item[{\rm (b)}] Space $(1+\gamma) nr$ bits, query time $O(1+\log(\frac1\gamma))$, 
expected construction time $O(n)$%
%, for any $\delta>0$
.
\end{description}
\end{theorem}

Our basic data structure and query evaluation algorithm is the same as in~\cite{CKRT04,p:MWHC96}.
The new contribution is to analyze a different construction algorithm (suggested in~\cite{HS94}) that is able to achieve a better space usage. Our analysis needs tools and theorems from linear algebra, while that of~\cite{CKRT04,p:MWHC96} is based on random graph theory (\cite{HS94} provided only experimental results). To get a data structure that allows expected linear construction time we devise a new variant of the data structure and query evaluation algorithm, retaining simplicity and non-adaptivity.

Our second contribution is to point out an intimate connection between the approximate membership problem and the retrieval problem:

\begin{theorem}\label{thm:reduction:approximate:retrieval}
Assuming free access to fully random hash functions, any static retrieval data structure can be used to implement an approximate membership data structure having false positive probability $2^{-r}$, with no additional cost in space, and $O(1)$ extra time. This reduction is near-optimal in the sense that it can be used to solve the approximate membership problem in space that is within $O(\log\log n)$ bits of optimal whp.
\end{theorem}

The papers on Bloom filters, and the papers on retrieval~\cite{CKRT04,p:MWHC96} all make the assumption of access to fully random hash functions, as in the above. We show how our data structures can be realized on a RAM, with a small additional cost in space:

\begin{theorem}\label{thm:with:split:and:share}
In the setting of Theorem~\ref{thm:main}, for some $\varepsilon > 0$, we can avoid the assumption of fully random hash functions and get data structures with the following space and time complexities\emph{:}
\begin{description}\setlength{\itemsep}{0pt}
\item[{\rm(a)}] Space $nr+O(n^{1-\varepsilon})$ bits, query time $O(\log n)$, expected construction time $O(n^{1+\gamma})$.
\item[{\rm(b)}] Space $(1+\gamma) nr$ bits, query time $O(1+\log(\frac1\gamma))$, expected construction time $O(n)$%,
%for any $\delta>0$
.
\end{description}
\end{theorem}

Our results have a couple of other implications in data structures. We improve the space usage of a recent simple construction of (minimal) perfect hashing of Botelho \emph{et al}.~\cite{BPZ07} (Section~\ref{sec:perfect:hashing}). In addition, we show a close relationship between ``cuckoo hashing''-like dictionaries and retrieval structures (Section~\ref{sec:hashing:retrieval}). This implies improved upper bounds on the space usage of $k$-ary cuckoo hashing~\cite{FPSS05} (or equivalently, of the 1-orientability threshold of a $k$-uniform random hypergraph with $n$ edges).

\subsection{Overview of paper}

Section~\ref{sec:retrieval} describes our basic retrieval data structure and its analysis, using a result due to Calkin~\cite{Cal97}. This leads to part (b) of Theorem~\ref{thm:main}, except that the construction time is $O(n^3)$. Part (a) is shown in Section~\ref{sec:retrieval:optimal:space}, using a result of Cooper~\cite{Coo00}. The reduction of approximate membership to retrieval, Theorem~\ref{thm:reduction:approximate:retrieval}, is presented in Section~\ref{sec:bloom}. Section~\ref{sec:linear-retrieval} completes the proof of part (b) of Theorem~\ref{thm:main}  by showing how the construction algorithm can be made to run in linear time. Section~\ref{sec:perfect:hashing} describes an application of our results to perfect hashing. Some issues, such as circumventing the full randomness assumption, 
leading to Theorem~\ref{thm:with:split:and:share}, are discussed in the appendices.

%%%%%%%%%%%%%%%%%%%%%%%%%%%%%%%%%%%%%%%%%%%%%%%%%%%%%%%%%%%%%%%%%%%%%%%%%%%%%%%%%%%%%%%%%%%%%%%%%%%%
\section{Retrieval in constant time and almost optimal space}\label{sec:retrieval}
%%%%%%%%%%%%%%%%%%%%%%%%%%%%%%%%%%%%%%%%%%%%%%%%%%%%%%%%%%%%%%%%%%%%%%%%%%%%%%%%%%%%%%%%%%%%%%%%%%%%
In this section, we give the basic construction
of a data structure for retrieval with constant time
lookup operation and $(1+\delta)nr$ space.
As a technical basis, we start with describing a result by Calkin~\cite{Cal97}
regarding the probability that 
0-1-matrices with sparse rows chosen randomly have full row rank.
%As an application, we remark that our space-efficient retrieval structure 
%makes it possible to improve on space bounds for 
%(simple) minimal perfect hash functions in the style of \cite{CKRT04,BPZ07}.

%---------------------------------------------------------------------------------------------------
\subsection{Calkin's results}\label{sec:Calkin}
%---------------------------------------------------------------------------------------------------

All calculations are over the field 
GF$(2)=\mathbb{Z}_2$ with
$2$ elements. 
We consider binary matrices $M=(p_{ij})_{1\le i \le n,0 \le j < m}$
with $n$ rows and $m$ columns.
If $M$ is such a matrix,
then row vector $(p_{i0},\ldots,p_{i,m-1})$ is called $p_i$,
for $1\le i \le n$.
%Note that if $n>m$ then the rows of $M$ must be linearly dependent.

\begin{theorem}[Calkin \protect{\cite[Theorem 1.2]{Cal97}}]
\label{thm:Calkin}
For every $k>2$ there is a constant $\beta_k<1$ such that the following 
holds\emph{:} 
Assume the $n$ rows $p_1,\ldots,p_n$ of a matrix $M$ are chosen at random 
from the set of binary vectors of length $m$ and weight \emph{(}number of $1$s\emph{)} exactly $k$. Then\emph{:}
\begin{itemize}\setlength{\itemsep}{0pt}
\item[{\rm (a)}] If $n/m \le \beta < \beta_k$,
                  then $\Pr(\mbox{$M$ has full row rank})\to1$ \emph{(}as $n\to \infty\emph{)}$.
\item[{\rm (b)}] If $n/m \ge \beta > \beta_k$,
                 then $\Pr(\mbox{$M$ has full row rank})\to0$ \emph{(}as $n\to \infty\emph{)}$.
\end{itemize}
Furthermore, $\beta_k - (1-(e^{-k}/(\ln2)) \to 0$ for $k\to\infty$ 
\emph{(}exponentially fast in $k$\emph{)}. 
\end{theorem}

\begin{remark}\label{remark:cycles:in:graphs}
\begin{rm}
%The case $k=2$ is omitted in this discussion.
It has been noted earlier in related work 
\cite{p:MWHC96}  
that the question whether a matrix 
with $m$ columns and randomly chosen 
rows of weight 2 has full row rank is equivalent
to the question whether the graph 
with $M$ as its vertex-edge incidence 
matrix is cyclic.
The threshold value for this case is $\beta_2=2$,
as is well known from the theory of random graphs.
In \cite{p:MWHC96} and \cite{BPZ07} it is 
explored how this fact can be used for constructing
perfect hash functions, in a way that 
implicitly includes the construction of 
retrieval structures.   
\end{rm}
\end{remark}

\begin{remark}\label{remark:convergence:speed}
\begin{rm}
A closer look into the proof of Theorem 1.2 in~\cite{Cal97}
reveals that for each $k$ there is some $\varepsilon=\varepsilon_k>0$
such that 
in the situation of Theorem~\ref{thm:Calkin}(a) 
%(below the threshold) 
we have
$\Pr(\mbox{$M$ has full row rank})=1-O(n^{-\varepsilon})$. 
% Rasmus: Does the big-O constant depend on k?
% Martin: Yes, definitely. 
% Otherwise it seems that this would imply something about the speed of convergence (in n).
% Don't think so. 
The following values are suitable: $\varepsilon_3=\frac27$, $\varepsilon_4=\frac57$,
$\varepsilon_k=1$ for $k\ge5$. 
\end{rm}
\end{remark}

According to \cite{Cal97}, the threshold value $\beta_k$ is characterized as follows:
Define
\begin{equation}\label{eq:10}
	f(\alpha,\beta)= -\ln 2 - \alpha\ln\alpha-(1-\alpha)\ln(1-\alpha) + \beta\ln(1+(1-2\alpha)^k),
\end{equation}
for $0< \alpha < 1$. Let $\beta_k$ be the minimal $\beta$ so that
$f(\alpha,\beta)$ attains the value 0 for some $\alpha\in(0,\frac12)$.
Using a computer algebra system, 
it is easy to find approximate values for small $k$, see Table~\ref{tab:ThresholdValues}.
This table also lists upper bounds for the reciprocals $\beta_k^{-1}$, 
since these are the figures we will utilize later. 
Calkin further proves that 
\begin{equation}\label{eq:417}
\beta_k = 1 - \frac{e^{-k}}{\ln 2} 
         - \frac{1}{2\ln2}\Bigl(k^2-2k+\frac{2k}{\ln2}-1\Bigr)\cdot e^{-2k} \pm O(k^4)\cdot e^{-3k},
\end{equation}
as $k\to\infty$. It seems that the approximation obtained by
omitting the last term in (\ref{eq:417}) is quite good already for small values of $k$.
(See the row for $\beta_k^{\rm appr}$ in Table~\ref{tab:ThresholdValues}.) 
%\begin{table}
%	\begin{center}
%		\begin{tabular}{c|c|c|c|c|c}
%			$k$ & 2 & 3 & 4 & 5 & 6 \\
%			\hline
%			$\beta_k$ & 0.5 & 0.88949 & 0.96714 & 0.98916 &0.99622\\
%			\hline
%			$\beta_k^{\rm appr}$ & 0.5 & 0.9091 & 0.9690 & 0.9893 &0.99624\\
%			\hline
%			$\beta_k^{-1}$ & 2 & 1.1243 & 1.034 & 1.011 & 1.0038
%		\end{tabular}
%	\end{center}
%	\caption{Threshold values}
%	\label{tab:ThresholdValues}
%\end{table}

\begin{table}
	\begin{center}
		\begin{tabular}{c||c|c|c|c}
			$k$  & 3 & 4 & 5 & 6 \\
			\hline\hline
			$\beta_k$  & 0.88949 & 0.96714 & 0.98916 &0.99622\rule[-5pt]{0pt}{18pt}\\
			\hline
			$\beta_k^{\rm appr}$  & 0.9091 & 0.9690 & 0.9893 &0.99624\rule[-5pt]{0pt}{18pt}\\
			\hline
			$\beta_k^{-1}$  & 1.1243 & 1.034 & 1.011 & 1.0038\rule[-5pt]{0pt}{18pt}
		\end{tabular}
	\end{center}
	\caption{Approximate threshold values from Theorem~\ref{thm:Calkin}, using (\ref{eq:10}) and (\ref{eq:417}).}
	\label{tab:ThresholdValues}
\end{table}

% \textbf{This is not yet clear, since using these results may be a little 
% inconvenient:}
% We will also need results from \cite{Cal96} regarding 
% matrices that have entries from GF$(q)$ for $q>2$.
% (It turns out that for each $q\ge 3$ and each $k\ge2$ 
% there is a constant $\beta_k(q)$, which plays the same role for $q$
% as $\beta_k$ does for $2$. 
% The constants $\beta_k(q)$ decrease as $q$ grows (asymptotically:
% $1-\beta_k(q)\sim \frac{q-1}{\ln q}\cdot e^{-k}$, for $q$ fixed and $k\to\infty$.)

\begin{remark}
\begin{rm}
Results similar to those of Calkin~\cite{Cal96,Cal97}, but for a different model,
were obtained independently by Balakin, Kolchin, and  Khokhlov
\cite{BKK92,Kol94,KK95}. Further results in a similar vein can be found in
a paper by Cooper~\cite{Coo99}.
\end{rm}
\end{remark}

%---------------------------------------------------------------------------------------------------
\subsection{The basic retrieval data structure}
%\subsection{Retrieval assuming full randomness}
\label{sec:retrieval:with:full:randomness}
%---------------------------------------------------------------------------------------------------

Now we are ready to describe a retrieval data structure.
Assume $f\colon S \to \{0,1\}^r$ is given, for a set $S=\{x_1,\ldots,x_n\}$.
For a given (fixed) $k\ge3$ let $1+\delta>\beta_k^{-1}$
be arbitrary and let $m=(1+\delta)n$.
We can arrange the lookup time to
be $O(k)$ and the number of bits in the data structure to be
$mr=(1+\delta)nr$ plus lower order terms.%
\footnote{For simplicity, 
in the notation we suppress rounding nonintegral values to a suitable near integer.} 
%For the analysis, we have to assume that $n$ is large enough. 

We assume that we have access to $k$ hash functions
with ranges $[m],\ldots,[m-k+1]$
that behave fully randomly on the keys of $S$, 
and that, in case the construction below fails, we may switch to 
a new independent set of $k$ hash functions,
again random on $S$. 
It is not hard to see that this assumption makes it possible
to define a mapping
$
U\ni x \mapsto A_x \in \dbinom{[m]}{k}
$,
where $\binom{X}{k}$ denotes the set of all subsets of $X$ with $k$ elements,
so that computing $A_x$ from $x\in U$ takes time $O(k)$, 
and so that $(A_x)_{x\in S}$
is fully random on $S$.
(For details see Appendix~\ref{sec:creating:sets}.)
We need to store a few bits to record which 
set of hash functions was used in the successful construction. 

The construction starts from $=\{x_1,\ldots,x_n\}$
and the bit strings $u_i=f(x_i)\in\{0,1\}^r$, $1\le i \le n$.
We consider the matrix
\begin{equation}\label{eq:500}
M=(p_{ij})_{1\le i \le n, 0\le j <m},\mbox{ with $p_{ij}= 1$ if $j\in A_{x_i}$ and $p_{ij}= 0$ otherwise.}
\end{equation}
Theorem~\ref{thm:Calkin}(a) (with Remark~\ref{remark:convergence:speed}) says that 
$M$ has full row rank with probability $1-O(n^{-\varepsilon})$ for some
$\varepsilon>0$. Assume $n$ is so large that this happens with probability
at least $\frac34$.

If $M$ does have full row rank, 
the column space of $M$ is all of $\{0,1\}^n$,
hence for all $u\in\{0,1\}^n$
there is some $a\in\{0,1\}^m$ with $M \cdot a=u$.
More generally, we arrange the bit strings $u_1,\ldots,u_n\in\{0,1\}^r$ 
as a column vector $u=(u_1,\ldots,u_n)^{\sf T}$.
We stretch notation a bit (but in a natural way)
so that we can multiply binary matrices with vectors of $r$-bit strings:
multiplication is just bit/vector multiplication and addition is bitwise XOR.
It is then easy to see, working with the components of the $u_i$ separately, 
that there is a (column) vector $a=(a_0,\ldots,a_{m-1})^{\sf T}$ with entries in $\{0,1\}^r$
such that $M\cdot a=u$.
---
We can rephrase this as follows
(using $\oplus$ as notation for bitwise XOR): 
For $a\in(\{0,1\}^r)^m$ and $x\in U$ define 
\begin{equation}\label{eq:405}
h_a(x)={\textstyle\bigoplus\limits_{j\in A_x}} a_j.
\end{equation}
Then for an arbitrary sequence $(u_1,\ldots,u_n)$ 
of prescribed values from $\{0,1\}^r$
there is some $a\in(\{0,1\}^r)^m$ with
$h_a(x_i) = u_i$, for $1\le i \le n$.
Such a vector $a\in(\{0,1\}^r)^m$,
together with an identifier for the set of
hash functions used in the successful construction,
is a data structure for
retrieving value $u_i=f(x_i)$, given $x_i$. 
There are $k$ accesses to the data structure,
plus the effort to evaluate $k$ hash functions on $x$ 
and calculate the set $A_x$ from $x$ (see Appendix~\ref{sec:creating:sets}).

\begin{remark}\label{remark:seiden:hirschberg}
\begin{rm}
A similar construction (over arbitrary fields GF$(q)$) 
was described by Seiden and Hirschberg~\cite{HS94}
%, for constructing an ordered minimal perfect hash function%
. 
However, those authors did not have
Calkin's results, and so could not
give theoretical bounds on the number $m$
of columns needed.
Also, our construction generalizes the 
approach of Chazelle \emph{et al.}~\cite{CKRT04,p:MWHC96}
who required that $M$ could be transformed into
echelon form by permuting rows and columns,
which is sufficient, but not necessary,
for $M$ to have full row rank.
\end{rm} 
\end{remark}

Some details of the construction are missing.
We describe one of several possible ways to proceed. 
--- From $S$, we first calculate the sets $A_{x_i}$, $1\le i \le n$, 
in time $O(n)$. 
Using Gaussian elimination, 
we can check whether the induced matrix $M=(p_{ij})$ has full row rank.
If this is not the case,
we start all over with a new set of $k$ hash functions, leading to new sets $A_{x_i}$.
This is repeated until a suitable matrix $M$ is obtained. 
The expected number of repetitions is $1+O(n^{-\varepsilon})$.
For a matrix $M$ with independent rows Gaussian elimination will
also yield a ``pseudoinverse'' of $M$, that is, 
an invertible $n \times n$-matrix $C$ 
(coding a sequence of elementary row transformations
without row exchanges) with the property that
in $C\cdot M$ the $n$ unit vectors occur as columns: 
\begin{equation}\label{eq:C:times:M}
	\mbox{$\forall i$, $1\le i \le n$, $\exists$  $b_i\in[m]$: column $b_i$ of $C\cdot M$ 
	equals $e_i^{\sf T}=(0,\ldots,0,1,0,\ldots,0)^{\sf T}$.} 
\end{equation}
Given $u=(u_1,\ldots,u_n)\in\{0,1\}^n$ we wish to 
find a solution $a\in\{0,1\}^m$ of the system
\begin{equation}\label{eq:retrieval:10}
(C\cdot M) \cdot a = C \cdot u = u'=(u_1',\ldots,u_n')^{\sf T}.
\end{equation} 
Since $C\cdot M$ has the unit vectors in columns $b_1,\ldots,b_n$, 
we can easily read off a special $a$ that solves~(\ref{eq:retrieval:10}):
Let $a_j=0$ for $j\notin \{b_1,\ldots,b_n\}$,
and let $a_{b_i}=u_i'$ for $1\le i \le n$.
Exactly the same formula works if $u$, $u'$, 
and $a$ are vectors of $r$-bit strings.
---
We have established the following.  

\begin{theorem}
Assume that fully random hash functions from keys $x\in S$
to ranges $[m],\ldots,[m-k+1]$ are available
\emph{(}with the option to choose such functions repeatedly
and independently\emph{)}. Let $k > 2$ be fixed,
and let $1 + \delta >\beta_k^{-1}$.

Then for $n$ large enough the following holds\emph{:}
Given $S=\{x_1,\ldots,x_n\}$ and 
a sequence $(u_1,\ldots,u_n)$ of prescribed elements in $\{0,1\}^r$,
we can find a vector  $a=(a_0,\ldots,a_{m-1})$ with elements in $\{0,1\}^r$
such that $h_a(x_i)=u_i$, for $1\le i \le n$.

The expected construction time is $O(n^3)$, the scratch space needed is $O(n^2)$.
\end{theorem}

\begin{remark}\begin{rm} 
We note that only $n$ entries of $a$ are significant, namely entries $b_1,\dots,b_n$. The other entries can be chosen arbitrarily, for example equal to {\bf 0}. This observation implies that there is an alternative data structure whose redundancy is independent of $r$. A constant time rank data structure (e.g.~Theorem 4.4 in~\cite{dict-jour}) can be used to identify the entries in $\{b_1,\dots,b_n\}$ and map them to entries in a ``compressed'' array of size $n$. The space usage of the rank data structure is within a lower order term of the entropy of the set $\{b_1,\dots,b_n\}$, which is $\log_2 \binom{m}{n}$. The drawback of this is that accesses to the compressed array are now adaptive, as they depend on lookups in the rank data structure.
\end{rm}
\end{remark}

\begin{remark}\begin{rm} 
At the first glance, the time complexity 
of the construction seems to be  forbiddingly large.
However, using a trick described in Appendix~\ref{sec:split:and:share}
(``split-and-share'')
makes it possible to obtain a data structure
with the same functionality and space bounds (up to a $o(n)$ term) in time
$O(n^{1+\gamma})$ for any given $\gamma>0$.
In Section~\ref{sec:linear-retrieval} we show how
to construct a retrieval structure 
with essentially the same space requirements in expected linear time.
\end{rm}
\end{remark}

%%%%%%%%%%%%%%%%%%%%%%%%%%%%%%%%%%%%%%%%%%%%%%%%%%%%%%%%%%%%%%%%%%%%%%%%%%%%%%%%%%%%%%%%%%%%%%%%%%
\section{A retrieval structure with optimal space}
\label{sec:retrieval:optimal:space}
%%%%%%%%%%%%%%%%%%%%%%%%%%%%%%%%%%%%%%%%%%%%%%%%%%%%%%%%%%%%%%%%%%%%%%%%%%%%%%%%%%%%%%%%%%%%%%%%%%

The purpose of this section is to prove Theorem~\ref{thm:main}(a), i\,e., to show that
there is a data structure supporting retrieval
that requires only space $nr+O(\log\log n)$ bits whp.~%
(note that $nr$ bits is a lower bound), and 
in which one retrieval operation takes logarithmic
time. More precisely, $O(\log n)$ table entries
are read and combined by XOR.
The idea is the following. 
We use the same setup as in Section~\ref{sec:retrieval:with:full:randomness},
excepting that the range size $m$ is equal to $n$, and that $k$, 
the size of the sets $A_x$, is 
chosen to be $\Theta(\log n)$.
We set up the matrix $M$ as in (\ref{eq:500}).
The rows correspond to the $n$ keys $x_1,\ldots,x_n$ in $S$,
the columns to the range $[n]$. 
In order to argue that the induced $n\times n$-matrix $M$
has full rank at least with constant probability, 
we wish to use the following theorem.
 
\begin{theorem}[Cooper \cite{Coo00}, Theorem~2(a)] 
\label{thm:cooper}
Let $M=(p_{ij})_{1\le i\le n,0\le j <n}$ 
be an $n\times n$-matrix filled with $0$s and $1$s in
the following way\emph{:}
The entries are chosen independently, 
and each entry $p_{ij}$ is $1$ with probability $p=p(n)=(c+1)\log (n)/n$,
and $0$ with probability $1-p$, where $c>0$ is an arbitrary constant.  Then
$
\lim\limits_{n\to\infty}\Pr(\mbox{$M$ is regular}) = c_2
$,
where $c_2=\prod\limits_{1\le i \le n}(1-2^{-i})\approx 0.28879$.
\end{theorem}
\textbf{Note.} $c_2$ is the probability that a random 0-1-matrix, 
each entry being 0 or 1 with probability $\frac12$, is regular. 
We will work with the constant $c=1$ or $p=2\log(n)/n$ throughout, 
but any other constant would do as well. 
The statement of the theorem in Cooper's paper is even more general. 

A slight difficulty arises
in that the number of $1$s in a row is fixed to 
be $k$ in the setup of Section~\ref{sec:retrieval:with:full:randomness},
and is binomially distributed in Cooper's theorem. 
%We could try to prove the analogue of the theorem 
%for rows of fixed logarithmic weight\footnote{A fixed {\em even\/} k would not work, 
%since then the matrix is not regular: The vector consisting of all $1$s is in the nullspace.},
%but we choose a different way: 
The idea to resolve this is as follows:
The size $k(x)$ of set $A_x$, for $x\in S$, is chosen
at random according to the binomial distribution.
Then the rows of $M$ will have weight $\Theta(\log n)$ with high 
probability, as noted in the following lemma,
which is easy to prove by Chernoff bounds, e.\,g.,~\cite[Theorems 4.4 and 4.5]{MU:05}.

\begin{lemma}\label{lem:not:too:many:ones}
In the situation of Theorem~\ref{thm:cooper}, with $c=1$,
the probability that $M$ has a row in which there are more than $4\log n$ or
fewer than $\frac12\log n$ $1$s is $1/n^{\Omega(1)}$.
\end{lemma}
%Rasmus: Why do we need the lower bound?

\subsection{Sampling the binomial distribution}

Using one extra hash function $q$ (with range $[n^3]$, say),
for each $x\in S$ we can choose a number $k(x)$ at random from the 
binomial distribution conditioned on $[\frac12\log n,4\log n]$.
Then Lemma~\ref{lem:not:too:many:ones} implies that
the deviation in probability 
from the situation of Theorem~\ref{thm:cooper} is $o(1)$.
Specifically, assume $n$  and $p=2\log(n)/n$ are given,
and that a fully random hash function $g$ with range $[n^4]$ is available.
We wish to sample from $\mathrm{B}(n,p)$ conditioned on
$[\frac12\log n, 4\log n]=[\frac14np,2np]$, at least approximately.
For this, we prepare a table of all $O(\log n)$ values $F(i)=\Pr(X\le i)$,
$\frac12\log n \le i \le 4\log n$, for the corresponding distribution function $F$. 
For given $g(x)$ we find $i$ with 
$\frac12(F(i-1) + F(i)) \le g(x)/n^4 <  \frac12(F(i) + F_{\rm B}(i+1))$,
and return $i$.
An easy calculation shows that both  
$\mathrm{B}(n,p;\frac14np)$ and $\mathrm{B}(n,p;\frac2np)$
are in $[n^{-1},n^{-2}]$, so that the error we make in comparison
to the true binomial distribution is in $O(1/n^2)$,
which for our purposes is negligible.  

Having fixed $k(x)$, $x\in S$,
using $4 \log n$ further hash functions with ranges $[n-\ell+1]$, $1\le \ell \le 4\log n$,
for each $x\in S$ a set $A_x=\{h_1(x),h_2(x),\ldots,h_{k(x)}(x)\}$
is chosen at random from $\binom{[n]}{k(x)}$.
(For details see Appendix~\ref{sec:creating:sets}.)

\subsection{Putting things together}

According to Theorem~\ref{thm:cooper} the resulting matrix $M$ will 
have full rank with probability $0.28879 \pm o(1)$.
If it turns out $M$ is singular, 
we start all over, with a new set of $1+4\log n$ hash functions. 
With $O(\log n)$  such trials, the probability
that all resulting matrices $M$ are singular 
can be made as small as $n^{-d}$ for an 
arbitrary constant $d$. 
The information which set of hash functions succeeds 
is part of the data structure. Recording this information
takes at most $O(\log\log n)$ bits.
The remaining details of the construction are
the same as in Section~\ref{sec:retrieval:with:full:randomness}.

Lookup works as follows:
Given a key $x\in U$, use $q$ to calculate $k(x)$.
If this happens to be outside the range $[\frac12\log n,4\log n]$,
return an arbitrary value. 
Otherwise calculate $A_x=\{h_1(x),h_2(x),\ldots,$ $h_{k(x)}(x)\}$ 
and return the value given by (\ref{eq:405}), with $k(x)$ in place of $k$.     

%\textbf{(How to construct the distribution?)}
%Let B$(n,p;\ge i)$ denote the probability that 
%a  B$(n,p)$-distributed random variable is $\ge i$.
%We work in rounds 1, 2, 3, \dots. 
%In round $i$ we use hash function $\bar h_{2i-1}$
%to determine whether $|A_x|$ should be $i$ or larger, 
%as follows: If
%$$
%\bar h_{2i-1}(x)/n \ge \frac{\mbox{B$(n,p;\ge i)$}}{\mbox{B$(n,p;\ge i-1)$}},
%$$
%then the value $\bar h_{2i}$ is used to construct an $i$th element of $A_x$
%(as before, check the reference). 
%Otherwise we decide that $|A_x|=i-1$.
%If size $12\log n$ has been reached, we check the next 
%function $h_{2i+1}(x)$ for $i=12\log n$. 
%If the value indicates that
%$A_x$ should be even larger, we stop ---
%by Lemma~\ref{lem:not:too:many:ones} this happens with very small probability.
%(To obtain an even better approximation to the binomial distribution, 
%one may use larger ranges like $[n^2]$ for the functions $h_{2i-1}$.)

%(We are missing an argument that the approximation does not hurt. 
%As far as I can see, the probability of ``error'' in every step is at most $1/n$, 
%so there should be no problem, at least if we increase the range a little.) 

This proves Theorem~\ref{thm:main}(a).
We do not try to reduce the cost $O(n^3)$ for 
solving the linear system.
Using the ``split-and-share'' trick described in Appendix~\ref{sec:split:and:share}
we may avoid the assumption of fully random hash functions 
and obtain a retrieval structure as described in 
Theorem~\ref{thm:with:split:and:share}(a).

% In full version, we may want to discuss 1 vs multiple hash functions (essentially equivalent), and also the assumption of 1 versus multiple ranges.

%%%%%%%%%%%%%%%%%%%%%%%%%%%%%%%%%%%%%%%%%%%%%%%%%%%%%%%%%%%%%%%%%%%%%%%%%%%%%%%%%%%%%%%%%%%%%%%%%%%%
\section{Approximate membership}\label{sec:bloom}
%%%%%%%%%%%%%%%%%%%%%%%%%%%%%%%%%%%%%%%%%%%%%%%%%%%%%%%%%%%%%%%%%%%%%%%%%%%%%%%%%%%%%%%%%%%%%%%%%%%%

We prove Theorem~\ref{thm:reduction:approximate:retrieval}. ---
Let $S\subseteq U$ with $|S|=n$ and an error bound $2^{-s}$ be given. 
We let $q\colon U\to [2^s]$ be a fully random hash function. 
Using the methods from Sections \ref{sec:retrieval:with:full:randomness}
resp.~\ref{sec:retrieval:optimal:space} 
we build a retrieval structure $D$
that associates value $q(x)$ with $x\in S$.
The space requirements for the data structures 
and the times for construction and retrieval are 
inherited from the corresponding retrieval structures. 

A query for $x\in U$ returns ``yes'' if $D(x)=q(x)$ and ``no'' otherwise. 
It is then clear that a query for an $x\in S$ always
yields ``yes''. 
A query for an $x\in U-S$ yields ``yes''
with probability $2^{-s}$, since $D(x)$ is given by
the data structure, which is 
determined by $q(y),y\in S$ (and some other random choices), 
and hence is independent of $q(x)$.

\begin{remark}\label{remark:bloomier:filter}
\begin{rm}
Of course, we may combine the retrieval data structure
of Section~\ref{sec:retrieval:with:full:randomness}
with the approximate membership data structure just described to obtain
a data structure that needs space $(1+\delta)n(r+s)$
and has the functionality of a ``Bloomier filter'' 
as described in 
\cite{CKRT04}: On a query for $x\in S$, a prescribed
value $f(x)\in \{0,1\}^r$ is returned, while for $x\in U-S$
the probability that some value from $\{0,1\}^r$ 
(and not some error symbol) is returned is $2^{-s}$. 
\end{rm}
\end{remark}

\begin{remark}\label{remark:approximate:with:split:and:share:one}
\begin{rm}
If we drop the assumption of fully random hash functions (like $q$)
being provided for free, only a $o(1)$ term
has to be added to the false positive probability. This is brief\/ly discussed in 
Remark~\ref{remark:approximate:with:split:and:share:two}
in Appendix~\ref{sec:split:and:share}.
\end{rm}
\end{remark}

%%%%%%%%%%%%%%%%%%%%%%%%%%%%%%%%%%%%%%%%%%%%%%%%%%%%%%%%%%%%%%%%%%%%%%%%%%%%%%%%%%%%%%%%%%%%%%%%%%%%%%%%%%
\section{Retrieval in almost optimal space, with linear construction time}
\label{sec:linear-retrieval}
%%%%%%%%%%%%%%%%%%%%%%%%%%%%%%%%%%%%%%%%%%%%%%%%%%%%%%%%%%%%%%%%%%%%%%%%%%%%%%%%%%%%%%%%%%%%%%%%%%%%%%%%%%

In this section we show how, using a variant of the retrieval 
data structure described in Section~\ref{sec:retrieval:with:full:randomness}, 
we can achieve linear expected construction time
and still get arbitrarily close to optimal space.
This will prove Theorem~\ref{thm:main}(b).
The reader should be aware that 
the results in this section hold asymptotically, 
only for rather large $n$.

Using the notation of Sections \ref{sec:Calkin}
and~\ref{sec:retrieval:with:full:randomness},
we fix some $k$ and some $\delta>0$ such that 
$(1+\delta)\beta_k>1$. 
Further, some constant $\varepsilon>0$ is fixed.
We assume that
$2k+1$ fully random hash functions
are at our disposal, with ranges we can choose, 
and in case the construction fails we can choose 
a new set of such functions, even repeatedly.  
(Appendix~\ref{sec:split:and:share} explains
how this can be justified.) 

Define $b=\frac12\sqrt{\log n}$.
We assume that $\varepsilon$ and $\delta$ are 
so small that $(1+\varepsilon)^2(1+\delta)<4$,
and hence that $b \cdot 2^{(1+\varepsilon)^2(1+\delta)b^2}=o(n/(\log n)^3)$.

%---------------------------------------------------------------------------------------------------
%\subsection{Constructing a retrieval structure in linear time}
%\label{subsec:linear:time:structure}
%---------------------------------------------------------------------------------------------------
Assume $f\colon S\to\{0,1\}^r$ is given,
the value $f(x)$ being denoted by $u_x$. 
The global setup is as follows: 
We use one fully random hash function $\varphi$ to map $S$ into the range $[m_0]$ with $m_0=n/b$.
In this way, $m_0$ blocks $B_i=\{x\in S\mid \varphi(x)=i\}$, $0\le i < m_0$, are created, 
each with expected size $b$.
The construction has two parts:
a primary structure and 
a secondary structure for the ``overflow keys''
that cannot be accommodated in the primary structure. 
This is similar to the global structure of a number of well-known dictionary implementations.
%, but there is a notable difference: We do not have space to store the keys of the set.
For the primary structure, we
try to apply the construction from Section~\ref{sec:retrieval:with:full:randomness}
to each of the blocks separately,
but only once, with a fixed set of $k$ hash functions.
This construction 
may fail for one of two reasons: (i) the block may be too large ---
we do not allow more than $b'=(1+\varepsilon)b$ keys in a block if it
is to be treated in the primary structure, or
(ii) the construction from  Section~\ref{sec:retrieval:with:full:randomness}
fails because the row vectors in the matrix $M_i$ induced by the sets $A_x$, $x\in B_i$,  
are not linearly independent. 
 
For the primary structure, 
we set up a table $T$ with $(1+\delta)(1+\varepsilon)n$
entries, partitioned into $m_0$ segments of size
$(1+\delta)(1+\varepsilon)b=(1+\delta)b'$.
Segment number $i$ is associated with block $B_i$.
%In addition, for each block $B_i$ a flag bit $v[i]$ is provided.
If the construction from Section~\ref{sec:retrieval:with:full:randomness}
fails, we set %
%$v[i]$ to 0, 
%and the keys from $B_i$ are accommodated in the secondary structure. 
all the bits in segment number $i$ to $0$ and use the secondary structure to associate keys in $B_i$ with the correct values. As secondary structure we choose a retrieval structure as in \cite{CKRT04,p:MWHC96},
built on the basis of a second set of, say, $3$ hash functions
(used to associate sets $A'_x\subseteq[1.3n']$ with 
the keys $x\in S'$) 
and a table $T'[0..1.3n'-1]$.
This uses space $1.3n'r$ bits, where $n'$ is the size of the set $S'$ of keys for which the construction failed (the ``\emph{overflow keys}'').
Of course, the secondary structure associates a value $f'(x)$ with {\em any\/} key $x\in S$. Rather than storing information about which blocks succeed we compensate for the contribution from $f'(x)$ as follows: If the construction 
succeeds for $B_i$,
we store $(1+\delta)b'$
vectors of length $r$ in segment number $i$ of table $T$ so that $x\in B_i$ is associated with the value $f(x)\oplus f'(x)$.
%and set $v[i]=1$.
On a query for $x\in U$, calculate $i=\varphi(x)$, then the offset
$d_i=(i-1)(1+\delta)b'$ of the segment for block $B_i$ in $T'$,  and return 
$$
\textstyle\bigoplus\limits_{j\in A_x} T[j+d_i] \;\oplus\; 
\textstyle\bigoplus\limits_{j\in A'_x}T'[j],
$$
%the multiplication being scalar-vector multiplication and 
$\oplus$ representing bitwise XOR in $\{0,1\}^r$.
It is clear that for $x\in S$ the result will be $f(x)$: For $x\in S'$ the two terms are {\bf 0} and $f(x)$, and for $x\not\in S'$ the two terms are $f'(x)$ and $f(x)\oplus f'(x)$.
%since the value $f(x)$ is encoded in the primary structure
%if $x\in S-S'$  and in the secondary  structure if $x\in S'$.
%and the correct value is selected by the bit $v[i]$.
Note that the accesses to the tables are
nonadaptive: all $k+3$ lookups may be carried out in parallel. In fact, if $T$ and $T'$ are concatenated this can be seen as the same evaluation procedure as in our basic algorithm~(\ref{eq:405}), the difference being that the hash functions were chosen in a different way (e.g., do not all have the same range).

Lemma~\ref{lem:sublinear:S:prime} below says that $\E(n')=o(n)$.
Before proving the lemma we conclude the space analysis assuming that it is true.
The overall space is 
$(1+\delta)(1+\varepsilon)n (r + 1/b) + c|S'|r \mbox{ bits}$
(apart from lower order terms). 
If $\gamma>0$ is given, we may choose $\varepsilon$ and $\delta$ 
(and $k$)
so that this bound is smaller than $(1+\gamma)nr$ for $n$ large enough.  

We proceed to show that that $|S'|=o(n)$ and how to achieve construction time $O(n)$.

\begin{lemma}\label{lem:sublinear:S:prime}
The expected number of overflow keys is $o(n)$.
\end{lemma}

\emph{Proof}:
It is sufficient to show 
that the expected number of blocks that have
more than $(1+\varepsilon)b$ keys or
for which the
construction from Section~\ref{sec:retrieval:with:full:randomness}
fails is $o(m_0)$.
Let $x_0\in S$ be an arbitrary key. 
The number of keys colliding with $x_0$ is
$\mathrm{B}(n-1,\frac1{m_0})$-distributed;
we may assume it is $\mathrm{B}(n,\frac1{m_0})$-distributed.
The expected number of keys that collide with $x_0$ under $\varphi$ then is $n/m_0\le b $. 
Since $b'=(1+\varepsilon)b$, a standard Chernoff bound (\cite[Thm.~4.4]{MU:05}) 
together with the fact that $b=\omega(1)$ yields that 
$
\Pr(x\mbox{ collides with more than $b'$ keys}) \le e^{-b\varepsilon^2/3}=o(1)
$.
Hence the expected number of keys in overfull blocks is $o(n)$.
Now assume a block $B_i$ has no more than $(1+\varepsilon)b$
keys, and the
construction from Section~\ref{sec:retrieval:with:full:randomness}
is applied (once).
There we noted, referring to Remark~\ref{remark:convergence:speed},
that the probability that the
$b'\times(1+\delta)b'$-matrix induced by the sets $A_x$, $x\in B_i$
will have linearly dependent rows
is $1/b^{\Omega(1)}$, 
which is $o(1)$ again. 
Hence the expected number of keys in blocks that create matrices
that do not have full rank is $o(n)$.  
$\Box$

\subsection{Matrix computations using tables}
                                                      \label{subsec:preliminaries:tables}
%%%%%%%%%%%%%%%%%%%%%%%%%%%%%%%%%%%%%%%%%%%%%%%%%%%%%%%%%%%%%%%%%%%%%%%%%%%%%%%%%%%%%%%%%%%%%%%%

As a building block in our construction algorithm for the primary structure we need some tables that allow us to do computations on small matrices efficiently. The details offer no surprises, but are given here for completeness.

Let $\varepsilon, \delta>0$ be such that $(1+\varepsilon)^2(1+\delta)<4$.
For given $n$ define $b=\frac12\sqrt{\log n}$,
and let $b'=(1+\delta)b$.
  We want to set up auxiliary tables that help in dealing
  with binary matrices $M$ with ${}\le b'$ rows and $(1+\delta)b'$ columns.  
There are no more than $b'\cdot 2^{(1+\varepsilon)^2(1+\delta)b^2}=o(n/(\log n)^3)$ 
such matrices.
For each such matrix $M$ we determinine and store whether its rows are independent; 
if this is the case we also calculate and store a pseudoinverse $C$, as described
in Section~\ref{sec:retrieval:with:full:randomness}.
The overall space needed for this table is $o(n)$ bits, and
the total time to calculate the entries is $o(n)$ steps, 
even if a simple method like Gaussian elimination is used. 

In Section~\ref{sec:retrieval:with:full:randomness}
we described what we mean by multiplying 
a matrix with a vector of bit strings.
For some integer constant $\ell\ge 3$
we prepare a table of all matrix-vector pairs $(L,v)$ and their product $L\cdot v$,
where the 0-1-matrix $L$ has ${}\le b'/\ell$ rows and $b'(1+\varepsilon)/\ell$
columns, and the vector $v$ has $b'(1+\varepsilon)/\ell$
entries that are bit strings of length $(\log n)/\ell$.
Such a table makes it possible to 
multiply a matrix with  ${}\le b'$ rows and $b'(1+\varepsilon)$ columns with a
vector whose entries are bit strings of length $O(\log n)$
in $O(\ell^3)$ word operations, hence in $O(1)$ time.
This table has size $o(n)$ bits as well. 

%%%%%%%%%%%%%%%%%%%%%%%%%%%%%%%%%%%%%%%%%%%%%%%%%

\subsection{Primary structure construction algorithm}

We are now ready to show the following lemma.
\begin{lemma}\label{lem:linear:time:primary}
The primary structure can be constructed in time $O(n)$.
\end{lemma}
\emph{Proof}:
It is clear that linear time is sufficient to find the blocks $B_i$ 
and identify the blocks that are too large. 
Now consider a fixed block $B_i$ of size at most $(1+\varepsilon)b$.
We must evaluate $|B_i|\cdot k$ hash functions
to find the sets $A_x$, $x\in B_i$,
and can piece together the matrix $M_i$ that is induced by these
sets in time $O(b)$ (assuming one can 
establish a word of $O(b)$ $0$s in constant time
and set a bit in such a word given by its position 
in constant time).  
%{\bf (R.: Should we have such machine model assumptions collected somewhere?)
% M.: No. Is clear as it is written.}
The whole matrix has fewer than $\log n$ bits and
fits into a single word. 
This makes it possible to use the precomputed tables described in Section~\ref{subsec:preliminaries:tables}.
Similarly, using precomputed tables 
we can check in constant time whether 
$M_i$ has linearly independent rows or not
and in the positive case find a pseudoinverse $C_i$.

Now assume a bit vector $u=(u_1,\ldots,u_{|B_i|})^{{\sf T}}\in\{0,1\}^{|B_i|}$, is given.
Using $C_i$ and a lookup table as in Section~\ref{subsec:preliminaries:tables} 
we can find $C_i\cdot u$ 
in constant time. 
A bit vector $a=(a_j)_{1\le j \le (1+\delta)b'}$ that solves $M_i\cdot a=u$ can then
be found in time $O(b)$.
This leads to an overall construction time of $O(n)$
for the whole primary structure. 

If the values in the range are bit vectors $f(x)=u_x\in \{0,1\}^r$, $x\in B_i$,
a construction in time $O(nr)$ follows trivially. 
We may improve this time bound as follows:
The lookup tables from  Section~\ref{subsec:preliminaries:tables} 
make it possible to multiply $C_i$ 
even with vectors of length up to $O(\log n)$ in constant time. 
This establishes a construction time of $O(n)$
for the general problem of representing a function
$f\colon S\to\{0,1\}^r$, using $r=O(\log n)$.
$\Box$

We have proved the following result, which is a precise 
and more general version
of Theorem~\ref{thm:main}(b):

\begin{theorem}
\label{thm:linear:time}
There is an algorithm ${\cal A}$ with the following properties.
For every $\gamma>0$ there is some $k=O(\log(1/\gamma))$ such that for 
all sufficiently large $n$ the following holds\emph{:}
Given a set $S\subseteq U$ with $n=|S|$ and a function $f\colon S\to \{0,1\}^r$,
with probability $1-o(1)$ algorithm ${\cal A}$ will
succeed in building a data structure $D$ such that\emph{:}
\begin{itemize}\setlength{\itemsep}{0pt}
	\item[{\rm(a)}] $D$ supports retrieval %\emph{(}i.e., calculating $f(x)$ from $x$\emph{)}
	in time $O(k)$, with no more than $2k+1$ hash function evaluations and $2k+1$
	\emph{(}nonadaptive\emph{)} 
	random accesses into tables storing $r$-bit vectors or bits\emph{;}
	\item[{\rm(b)}] the space occupied by $D$ is no more than $(1+\gamma)nr$ bits\emph{;}
	\item[{\rm(c)}] ${\cal A}$ runs in time $O(n)$.
\end{itemize}
\end{theorem}

%%%%%%%%%%%%%%%%%%%%%%%%%%%%%%%%%%%%%%%%%%%%%%%%%%%%%%%%%%%%%%%%%%%%%%%%%%%%%%%%%%%%%%%%%%%%%%%%
\section{Retrieval and dictionaries by balanced allocation}
                                                                           \label{sec:hashing:retrieval}
%%%%%%%%%%%%%%%%%%%%%%%%%%%%%%%%%%%%%%%%%%%%%%%%%%%%%%%%%%%%%%%%%%%%%%%%%%%%%%%%%%%%%%%%%%%%%%%%

In several recent papers,
the following scenario for (statically) storing a set 
$S\subseteq U$ of keys was studied.
A set $S=\{x_1,\ldots,x_n\}\subseteq U$
is to be stored in a table $\texttt{T[}0..m-1\texttt{]}$ 
of size $m=(1+\delta)n$ as follows:
To each key $x$ we associate a set $A_x\subseteq[m]$ 
of $k$ possible table positions.
Assume there is a mapping $\sigma\colon \{1,\ldots,n\}\to[m]$
that is one-to-one and satisfies $\sigma(i)\in A_{x_i}$, for $1\le i \le n$.
(In this case we say $(A_x,x\in S)$ is \emph{suitable} for $S$.)
Choose one such mapping and store $x_i$ in $\texttt{T[}\sigma(i)\texttt{]}$.
Examples of constructions that follow this scheme are 
cuckoo hashing~\cite{cuckoo-jour}, $k$-ary cuckoo hashing~\cite{FPSS05},
blocked cuckoo hashing~\cite{DieWei07,panigrahy05},
and perfectly balanced allocation~\cite{CRS03}.
In \cite{Cain07,fern07} threshold densities for blocked cuckoo
hashing were determined exactly.
These schemes are the most space-efficient dictionary
structures known, among schemes that store the keys explicitly in a hash table. 
For example, $k$-ary cuckoo hashing~\cite{FPSS05}
works in space $m=(1+\varepsilon_k)n$ with $\varepsilon_k=e^{-\Theta(k)}$.
Perfectly balanced allocation~\cite{CRS03}
works in optimal space $m=n$ with $A_x$ 
consisting of 2 contiguous segments of $[n]$ of length $O(\log n)$ each\footnote{The result of~\cite{CRS03} only directly applies when $n$ is divisible by the segment length. However, from the proof it is clear that one also gets a perfectly balanced allocation in the case where segment sizes may differ by~$1$.}.
% A pedantic comment: Their theorem only treats the case where all buckets have the same size, i.e., 
% technically we must assume that n is divisible by a number of size O(log n)...
% No doubt their theorem still holds if buckets have different capacities, but it is not clear to me if this 
% is obvious from the paper.

Here, we point out a close relationship between dictionary structures of this kind
and retrieval structures for functions $f\colon S\to R$, whenever the range $R$ is not too small. 
We will assume that $R=\mathbb{F}$ for a finite field $\mathbb{F}$ with $|\mathbb{F}|\ge n$.
(Using a simple splitting trick, 
this condition can be attenuated to  $|\mathbb{F}|\ge n^\delta$, see Appendix~\ref{sec:split:and:share}.)

From Section~\ref{sec:retrieval:with:full:randomness}
we recall equation~(\ref{eq:500}) where the matrix 
$M=(p_{ij})_{%
%\begin{subarray}{l}
1\le i \le n\\ 0\le j <m
%\end{subarray}%
}$ 
was defined from the 
sets $A_x$, $x\in S$.

\begin{observation}
Let $\mathbb{F}$ be an arbitrary field. If the $1$s in $M$ can be replaced by elements of $\mathbb{F}$
in such a way that the resulting matrix $M'=(p_{ij}')$ has full row rank over $\mathbb{F}$,
then $(A_x,x\in S)$ is suitable for $S$. 
\end{observation}
\begin{proof}
If $M'$ has full row rank, it has an $n\times n$ submatrix $N$ with nonzero determinant.
By the definition of the determinant there must be a mapping $\sigma\colon \{1,\ldots,n\}\to[m]$
with $\prod p'_{i\sigma(i)}\not=0$, hence $ p_{i\sigma(i)}=1$ for $1\le i \le n$.
\end{proof}

The observation implies that Calkin's bounds from Section~\ref{sec:retrieval} give upper space bounds for dictionary constructions like $k$-ary cuckoo hashing. We note that the figures from Table~\ref{tab:ThresholdValues} coincide nicely with the space bounds found by experiments in~\cite{FPSS05}, and are much tighter than the theoretical results found there.
Surprisingly, for fields that are not too small, the observation also works the other way around:
existence of a dictionary implies existence of a retrieval structure.    

\begin{theorem}
\label{thm:hashing:implies:retrieval}
   Assume a mapping $x\mapsto A_x$ is given that is suitable for $S$. 
   Then the following holds\emph{:}
   If $g_1,\ldots,g_k \colon S\to \mathbb{F}$ are random, 
	              then with  probability at least $1 - \frac{n}{|\mathbb{F}|}$
	               the matrix
\begin{equation}\label{eq:5000}
  M'=(p_{ij}')_{%
%\begin{subarray}{l}
1\le i \le n\\ 0\le j <m
%\end{subarray}%
},\mbox{ with $p_{ij}'= g_\ell(x_i)$ 
  if $j=h_\ell(x_i)$ and $p_{ij}'= 0$ otherwise}
\end{equation}
has full row rank over $\mathbb{F}$.
\end{theorem}
\begin{proof}
Define an auxiliary matrix  $P=(s_{ij})$ of 
polynomials in variables $X_{ij}$, $1\le i\le n$,
by  
$s_{ij}=X_{ij}$ if $j=h_\ell(x_i)$
 and  $s_{ij}=0$ for all other $j$'s.
Then the submatrix $P'$ of $P$ that corresponds to
the columns $\sigma(1),\ldots,\sigma(n)$ of $P$ satisfies
$$
\det(P')\not= \mbox{the $0$-polynomial.}
$$
This is because the mapping $i\mapsto \sigma(i)$
yields one term in the expansion of $\det(P')$ as a sum that does not vanish, 
and because the terms in this sum  
cannot cancel in any way. 
By the Schwartz-Zippel Theorem (see e.g.~\cite{MotRag95})
we know that
if we substitute random elements $g_\ell(x_i)$ from $\mathbb{F}$
for the variables $X_{ij}$ in $P'$, the probability
that the resulting matrix  $P'[X_{ij}|g_\ell(x_i)]$ (with $j=h_\ell(x_i)$) 
is regular is
at least 
$$
1-\deg(\det(P'))/|\mathbb{F}|\ge 1-n/|\mathbb{F}| \enspace .
$$
Clearly, if $P'[X_{ij}|g_\ell(x_i)]$ has full rank, 
the extension $P[X_{ij}|g_\ell(x_i)]=M'$ has full row rank.
\end{proof}
%$\Box$

The theorem implies the following: 
If the  mapping $x\mapsto A_x$ is suitable for $S$,
if $|\mathbb{F}|\ge2n$,
and if we have hash functions $g_1,\ldots,g_k \colon U\to \mathbb{F}$
that are random on $S$, then with probability at least $\frac12$ 
we can build a retrieval structure
for a function $f\colon S\to \mathbb{F}$ 
consisting of a table $T[0..m-1]$ with entries from $\mathbb{F}$
with $f(x)=\sum_{1\le \ell\le k}g_\ell(x)\cdot T[h_\ell(x)]$.
If we can switch to new functions $g_1,\ldots,g_k$ if necessary,
this construction succeeds in $O(\log n)$ iterations whp.

For example, from the dictionary constructions in~\cite{DieWei07},
or~\cite{CRS03}, resp., we obtain 
 retrieval structures with a table of 
size $\le(1+e^{-k})n$ and lookup time $O(k)$,
or optimal size $n$ and lookup time $O(\log n)$, resp.
In both cases for one retrieval operation we need to access only two 
contiguous segments of the table~$T$,
which makes these implementations very cache-friendly. 
Another, even more cache friendly, possibility is to do linear probing in a table of size $(1+\Omega(1))n$, in which case the lookup time can be bounded by $O(\log n)$ with high probability.
\footnote{We were not able to find this fact in the literature, but from Chernoff bounds it is easy to see that with high probability any interval of length $i\geq C\log n$, with a suitably large constant $C$ depending on $\alpha$, contains at most $i$ hash values. This implies that the maximum probe length is bounded by $i$.}.

%%%%%%%%%%%%%%%%%%%%%%%%%%%%%%%%%%%%%%%%%%%%%%%%%%%%%%%%%%%%%%%%%%%%%%%%%%%%%%%%%%%%%%%%%%%%%%%%
%\section{Other applications}
%                                                                           \label{sec:otherapps}
%%%%%%%%%%%%%%%%%%%%%%%%%%%%%%%%%%%%%%%%%%%%%%%%%%%%%%%%%%%%%%%%%%%%%%%%%%%%%%%%%%%%%%%%%%%%%%%%

%-----------------------------------------------------------------------------------------------
\section{Application to perfect hashing}\label{sec:perfect:hashing}
%-----------------------------------------------------------------------------------------------

The problem of perfect hashing is the following: Given $S\subseteq U$, 
construct a data structure that makes 
it possible to calculate $h(x)\in[m]$ for $x\in U$ (fast),
so that $h$ is 1-to-1 on $S$.
We show that the constructions of Chazelle \emph{et al.}~\cite{CKRT04} and  
Botelho \emph{et al.}~\cite{BPZ07}, resp., 
which lead to space-efficient representations that can be evaluated fast,
can be enhanced so as to achieve a smaller space bound. 

As before, we assume we have $k$ hash functions
with ranges $[m],\ldots,[m-k+1]$ that are fully random on $S$.
As described in Appendix~\ref{sec:creating:sets} this yields a mechanism 
to calculate an ordered set $A_x=\{h_1(x),\ldots, h_k(x)\}$ of distinct values in $[m]$,
for $x\in U$, so that $A_x$, $x\in S$, is fully random. 
This set system can be regarded as a hypergraph. 
Matrix $M=(p_{ij})_{1\le i \le n, 0\le j <m}$ 
from (\ref{eq:500}) is the vertex-edge incidence matrix of this hypergraph. 

Chazelle \emph{et al.}~\cite{CKRT04} (implicitly) and  Botelho \emph{et al.}~\cite{BPZ07}
have proposed a method for 
obtaining a perfect hash function in this situation. 
If $M$ can be brought into echelon form
by permuting rows and columns, 
then each row has 
a distinguished column,
namely the column 
in which the first 1 from the left in this row is positioned.
The column index is an element of $A_x$.
Thus, we obtain a one-to-one-mapping $\varphi$ from
$\{1,\ldots,n\}$ into  $\{0,\ldots,m-1\}$
so that $\varphi(i)\in A_{x_i}$.
This means that for each $i$ 
(or $x_i$) there is a unique $\lambda(i)\in \{0,\ldots,k-1\}$
such that $\varphi(i) = h_{\lambda(i)}(x)$.
If we can set up a data structure that makes
it possible to calculate $\lambda(i)$ from $x_i$,
we can calculate a one-to-one mapping from $S$ into 
$\{0,\ldots,m-1\}$  ---  a perfect hash function.
But providing such a data structure is in essence the retrieval problem!
We know how to solve this with a data structure with 
$\lceil\log k\rceil\cdot m$ bits.
In the special case of a matrix
that can be brought into echelon form by permuting 
rows and columns 
the system of equations can even be solved in linear time. 

Chazelle \emph{et al.}~\cite{CKRT04} and  Botelho \emph{et al.}~\cite{BPZ07}
gave different arguments why for $m$ sufficiently large matrix $M$
should (with high probability) 
admit a transformation into upper triangular form by row and column permutations. 
 Chazelle et al.~\cite{CKRT04} gave ad hoc calculations
 leading to the bound $m>kn$. 
Botelho et al.~\cite{BPZ07} used results from random (hyper)graph theory
to state much smaller bounds (no closed formula): 
$m>1.222n$ for $k=3$ and $m>1.295n$ for $k=4$, and so on, 
with the constants growing for growing $k$.
(The question asked in~\cite{BPZ07} is whether the 
hypergraph given by $A_x$, $x\in S$, is ``acyclic''.)
We avoid the use of random graph theory and resort to
Calkin's theorem (Theorem~\ref{thm:Calkin}) to show that
the bounds $\beta_k^{-1}$ from Table~\ref{tab:ThresholdValues}
are relevant for this situation as well. 
The disadvantage of our approach is
that the algorithms that construct the data
structure need more time, since they involve Gaussian elimination.
Again, the splitting trick from Appendix~\ref{sec:split:and:share}
can alleviate this problem.
 
Assume the matrix $M$ has
full row rank. We first calculate a pseudoinverse $C$ that satisfies eq. (\ref{eq:C:times:M}) in Section~\ref{sec:retrieval:with:full:randomness}.
Since columns $b_1,\ldots,b_n$ of $C\cdot M$ 
form a regular quadratic matrix,
and $C\cdot M$ is obtained from $M$ only by row transformations, 
columns  $b_1,\ldots,b_n$ of $M$ also 
form a regular matrix.
This means that the determinant of the
submatrix of $M$ formed by these columns
is nonzero --- hence, by the definition
of the determinant as a sum of products over all permutations, 
there must be a bijection
$\varphi\colon \{1,\ldots,n\}\leftrightarrow\{b_1,\ldots,b_n\}$
with $p_{i,\varphi(i)}\not=0$, hence $p_{i,\varphi(i)}=1$, for $1\le i \le n$.
%(We are working in the field $\mathbbm{Z}_2$.)
This means that $\varphi(i)\in A_{x_i}$, for $1\le i \le n$.
The mapping $\varphi$ may be found by an efficient
algorithm to calculate perfect matchings in bipartite graphs.
For each $i$, from $\varphi(i)$ we obtain a value
$\lambda(i) \in \{1,\ldots,k\}$ such that
$\varphi(i)=h_{\lambda(i)}(x_i)$.

We form a vector $(u_1,\ldots,u_n)$ by 
defining $u_i$ to be (the binary representation of)
$\lambda(i)-1$, using $r=\lceil \log k\rceil$ bits.
Applying the construction from Section~\ref{sec:retrieval:with:full:randomness}
we find a vector 
$a=(a_0,\ldots,a_{m-1})$ with elements in $\{0,1\}^r$
such that $h_a(x_i)=\lambda(i)-1$, for $1\le i \le n$.

Then the function
$$
h\colon U \to \{0,\ldots,m-1\}\, , \;\; x\mapsto   h_{h_a(x)+1}(x)
$$
is a perfect hash function for $S$ with range $\{0,\ldots,m-1\}$.
Evaluating the function amounts to calculating $h_a(x)$
as given by (\ref{eq:405}).
The function $h$ is represented by the table that contains the components of $a$.
This takes $m=(1+\delta)n$ words of $\lceil\log k\rceil$ bits,
where $\delta > \beta_k^{-1}-1$ is arbitrary.

Since $\beta_k^{-1} \sim 1 + e^{-k}/\ln 2$ for $k\to\infty$,
the relative space overhead $\delta$ may be made as small as we wish,
at the cost of larger $k$. 
A particularly attractive choice is $k=4$.
Since $\beta_4^{-1} < 1.035$, we could choose $m=1.035n$
and spend $2m$ bits for the representation 
of the vector $a$, which amounts to 
space requirements of $2.07n$ bits.

Bothelho \emph{et al.}~\cite{BPZ07}
describe how a perfect hash function may be turned 
into a minimal perfect hash function. 
There are several plausible techniques for this, 
one of them as follows:
One stores the set of locations in $\{0,\ldots,m-1\}-h(S)$ in a succinct rank
data structure~\cite{rankdict}. This table requires additional
space of 
%\textbf{Please check. I don't recognize this formula.}
%$0.035 n \cdot \ln(1.035/0.035) + o(n) \approx 0.171\ldots n + o(n)$ bits.
$0.035 n \cdot \log_2(1.035/0.35) + n \cdot \log_2 (1.035/1) \approx 0.22 n + o(n)$ bits.
The total space needed for the minimal perfect hash function is 
%$2.242n + o(n)$ bits,
$2.29n + o(n)$ bits,
which is a little better than the $2.7n$ from \cite{BPZ07}.
The price to pay for this improvement is that
to find the vector $a$ we must solve a system of linear equations
and solve an instance of the perfect matching problem,
while in \cite{BPZ07} a very simple linear-time algorithm is sufficient.
There is no big difference in the time needed to evaluate the
(minimal) perfect hash function.

%%%%%%%%%%%%%%%%%%%%%%%%%%%%%%%%%%%%%%%%%%%%%%%%%%%%%%%%%%%%%%%%%%%%%%%%%%%%%%%%%%%%%%%%%%%%%%%%%%%%

\section{Open problems}

Our construction relies on either the assumption that the hash functions are fully random, or on the split-and-share construction. An obvious open problem (shared with most data structures that use multiple hash functions) is whether simpler hash functions can do the same job.

Another question is to which extent the correspondence shown in Theorem~\ref{thm:hashing:implies:retrieval} of Section~\ref{sec:hashing:retrieval} also holds for small values of $r$. For example, is the space threshold for $k$-ary cuckoo hashing identical to Calkin's threshold for random matrices with $k$ $1$s in each row?  Similarly, if we imitate blocked cuckoo hashing~\cite{DieWei07} by restricting the sets $A_{x_i}$ to be a subset of the union of two intervals of size $k$ (depending on $x_i$), what is the best space usage we can get?

\medskip

\textbf{Acknowledgement.}
The authors thank Philipp Woelfel for several motivating discussions
on the subject.

%%%%%%%%%%%%%%%%%%%%%%%%%%%%%%%%%%%%%%%%%%%%%%%%%%%%%%%%%%%%%%%%%%%%%%%%%%%%%%%%%%%%%%%%%%%%%%%%%%%%

%%%%%%%%%%%%%%%%%%%%%%%%%%%%%%%%%%%%%%%%%%%%%%%%%%%%%%%%%%%%%%%%%%%%%%%%%%%%%%%%%%%%%%%%%%%%%%%%%%%%

%%%%%%%%%%%%%%%%%%%%%%%%%%%%%%%%%%%%%%%%%%%%%%%%%%%%%%%%%%%%%%%%%%%%%%%%%%%%%%%%%%%%%%%%%%%%%%%%
\appendix
%%%%%%%%%%%%%%%%%%%%%%%%%%%%%%%%%%%%%%%%%%%%%%%%%%%%%%%%%%%%%%%%%%%%%%%%%%%%%%%%%%%%%%%%%%%%%%%%

%%%%%%%%%%%%%%%%%%%%%%%%%%%%%%%%%%%%%%%%%%%%%%%%%%%%%%%%%%%%%%%%%%%%%%%%%%%%%%%%%%%%%%%%%%%%%%%%%%%%
\section{Creating random sets of size $k$ without repetitions}
                                                                           \label{sec:creating:sets}
%%%%%%%%%%%%%%%%%%%%%%%%%%%%%%%%%%%%%%%%%%%%%%%%%%%%%%%%%%%%%%%%%%%%%%%%%%%%%%%%%%%%%%%%%%%%%%%%%%%%

We brief\/ly justify the assumption
that given $k$ fully random hash functions
\emph{with ranges we can choose} there is a way to 
map each key $x$ to a fully random sequence (or ordered set) $A_x=(h_1(x),\ldots,h_k(x))$ 
with all different values in $[m]$.
Just take $k$ fully random hash functions $g_1,\ldots,g_k$
where $g_\ell$ has range $[m-\ell+1]$, for $\ell=1,\ldots,k$. 
The existence of the sequence $A_x$ 
is then easily proved:
For $\ell=1,\ldots,k$, 
let $h_\ell(x)$ be element number $g_\ell(x)$ in the set
$[m]-\{h_1(x),\ldots,h_{\ell-1}(x)\}$.
Algorithmically, it is simpler to work with an array
$\texttt{B}[0..m-1]$, initialized with  $\texttt{B}[j]=j$.
Then, sequentially for $\ell=1,\ldots,k$,
the value at position $\texttt{B}[g_{\ell}(x)]$
is exchanged with the value at $\texttt{B}[m-\ell]$.
The output sequence is $(h_1(x),\ldots,h_k(x))=(\texttt{B}[m-1],\ldots,\texttt{B}[m-k])$.
Clearly, this is a random sequence of $k$
distinct elements of $[m]$.
If space is an issue, one might not actually use this array $\texttt{B}$,
but just simulate the effect.
Time $O(k\log k)$ (using a search tree)
or expected time $O(k)$ (using hashing) is definitely sufficient.
If the ``split-and-share'' approach of Appendix~\ref{sec:split:and:share} is employed,
the space used for array $\texttt{B}$ is $O(n^\varepsilon)$.

%%%%%%%%%%%%%%%%%%%%%%%%%%%%%%%%%%%%%%%%%%%%%%%%%%%%%%%%%%%%%%%%%%%%%%%%%%%%%%%%%%%%%%%%%%%%%%%%%%%%
\section{How to circumvent the full randomness assumption:\\ ``split-and-share''} 
                                                                         \label{sec:split:and:share}
%%%%%%%%%%%%%%%%%%%%%%%%%%%%%%%%%%%%%%%%%%%%%%%%%%%%%%%%%%%%%%%%%%%%%%%%%%%%%%%%%%%%%%%%%%%%%%%%%%%%

While it is quite common to assume that fully random hash functions
are available for free in the context of Bloom filters and similar
data structures, in the realm of dictionary implementations
and construction of perfect hash functions
one prefers to use randomization in the algorithm, viz., universal hashing. 
We brief\/ly sketch how in the context of the static
data structures of this paper
universal hashing may be used to justify the full randomness assumptions. 
(For details, see~\cite{Die07}.)

We assume the reader is familiar with the concept of a universal class of hash functions from 
$U$ to $M$ (for distinct keys $x,y\in U$ and $h$ chosen at random
from the class we have $\Pr(h(x)=h(y)\le 1/|M|$) as well as the concept
of $k$-wise independent hash classes (for arbitrary distinct keys 
$x_1,\ldots,x_k\in U$ and $h$ chosen at random from the class
the hash values $h(x_1),\ldots,h(x_k)$ are fully random in $M$).
Constructions of such classes, with arbitrary ranges $M=[m]$, are well known,
see~\cite{CW:79,WC:79}. 

Let $S\subseteq U$, $|S|=n$, be fixed. It is a common idea to 
use a hash function $h\colon U\to [m]$ to split
$S$ into ``chunks'' $S_i=\{x\in S\mid h(x)=i\}$, for $i\in [m]$,
and then work on the ``chunks'' separately. 
In our context, we would {e.\,g.} construct a retrieval structure
for each $S_i$ separately in a dedicated segment of memory.
A more recent idea~\cite{Die07,DieWei07}
is to use a ``shared'' table of random words to
simulate hash functions that are fully random on each single chunk.
The total space usage of this table may be kept at $o(n)$.

If $m\ge 2n^{2/3}$ and $h\colon U\to[m]$
is chosen from a 4-universal class, 
then (as a standard calculation shows) 
the largest chunk $S_i$ will have at most $\sqrt{n}$ elements
with probability larger than $\frac34$.
We repeat choosing such $h$'s until one is found
that satisfies $\max\{|S_i| \mid 0\le i < m\}\le \sqrt{n}$.
This we fix and call it $h^0$.
The size of $S_i$ is called $n_i$.

It is a folklore fact that using $n_i^{1+\epsilon}$ space it is
very easy to provide a data structure that gives
fully random hash functions on $S_i$,
which can be evaluated in constant time.
A concrete construction might look as follows.

\begin{lemma}[{e.\,g.} \cite{Die07}]\label{lem:sim:randomness}
Let $r=2n^{3/4}$, and let $S'\subseteq U$ with $n'=|S'|\le\sqrt{n}$.
Given a $1$-universal class of hash functions from $U$ to $[r]$, 
we may in expected time $O(|S'|)$ find two functions $h_0,h_1$ from that class
such that if the two tables $\texttt{T}_0[0..r-1]$ and $\texttt{T}_1[0..r-1]$
are filled with random numbers from $[t]$
then $h'(x)=(\texttt{T}_0[h_0(x)]+\texttt{T}_1[h_1(x)])\bmod t$
defines a function that is fully random on $S'$.
\end{lemma}

This simple observation may be used as follows. 
For each of the $m$ chunks $S_i$ we  find and fix and store two functions
$h_{i,0},h_{i,1}$ (constant description size) 
as in Lemma~\ref{lem:sim:randomness}.
Further, we provide a suitable number $L$ of pairs of tables 
$\texttt{T}_{j,0}$ and $\texttt{T}_{j,1}$, $1\le j \le L$,
independently filled with random numbers from $[t]$
(we may even use varying ranges $[t_j]$).
Then 
$$
h'_{i,j}(x)=(\texttt{T}_{j,0}[h_{i,0}(x)] + \texttt{T}_{j,1}[h_{i,1}(x)])\bmod t \; 
                                           \mbox{, for \ } 1\le j \le L,
$$
provides $L$ fully random hash functions on $S_i$.
The total space taken up by the tables is no more than $2n^{3/4}\cdot L$ numbers,
so any $L=O(n^{1/4-\eta}/\log(t))$ will lead to a space usage of $O(n^{1-\eta})$ bits. 

The central observation is that this setup makes it possible
to work with $L$ truly fully random hash functions on $S_i$
to construct a data structure $D_i$ for solving the 
respective problem (retrieval, approximate membership) for the keys in $S_i$.
As long as we keep the data structures disjoint, 
the fact that the tables are shared is harmless.
More details may be found in~\cite{Die07}. 

\begin{remark}\label{remark:approximate:with:split:and:share:two}
\begin{rm}
In Remark~\ref{remark:approximate:with:split:and:share:one}
we mentioned a subtle problem in the false positive probability in
the approximate membership problem in case we use the split-and-share 
approach to simulate the hash functions. 
The problem is caused by the fact that a query point $y\in U-S$
is not in $S_i$ for $i=h(y)$ and hence there is no guarantee that
$q$ (simulated by $h_{i,0},h_{i,1}$ and 
some pair $\texttt{T}_{j,0}$ and $\texttt{T}_{j,1}$)
is fully random on $S_i\cup\{y\}$.
By the proof of Lemma~\ref{lem:sim:randomness} as given in \cite{Die07} 
it can be seen that
the probability that this happens for any fixed $y$ is $O(1/\sqrt{n})$.
This term has to be added to the false positive probability. 
\end{rm}
\end{remark}

%%%%%%%%%%%%%%%%%%%%%%%%%%%%%%%%%%%%%%%%%%%%%%%%%%%%%%%%%%%%%%%%%%%%%%%%%%%%%%%%%%%%%%%%%%%%%%%%%%%%
\section{Space lower bound for approximate membership}
                                                                            \label{app:approx-lower}
%%%%%%%%%%%%%%%%%%%%%%%%%%%%%%%%%%%%%%%%%%%%%%%%%%%%%%%%%%%%%%%%%%%%%%%%%%%%%%%%%%%%%%%%%%%%%%%%%%%%

For completeness we prove a lower bound for the space needed for an
approximate membership data structure. The proof is a slight extension
of the lower bound of Carter et al.~\cite{MR80h:68037}.

\begin{theorem}[Carter {\em et al}.~\cite{MR80h:68037}]
Let $u=|U|$ and consider an approximate membership data structure for
sets of size $n\geq 1$ and false positive probability~$\varepsilon$,
$0<\varepsilon<1$. Then the space usage in bits must be at least
$n\log_2 (1/\varepsilon)-O\left(\frac{(1-\varepsilon)n^2}{\varepsilon
u+ (1-\varepsilon) n}\right)$. Specifically, for $u>n^2/\varepsilon$
the space usage must be at least $n\log_2 (1/\varepsilon)-O(1)$ bits.
\end{theorem}

\emph{Proof}:
Any instance $\mathcal{I}$ of the data structure corresponds to a
subset $U_{\mathcal{I}} \subseteq U$, namely the set of elements $x$
for which the data structures answers ``$x\in S$''. For any set
$S\subseteq U$ there must be some instance $\mathcal{I}$ for which
$S\subseteq U_{\mathcal{I}}$. Furthermore, there must exist such an
instance where $|U_{\mathcal{I}}|\leq \varepsilon (u-n) + n$. This is
because the expected number of false negatives in $U - S$ is
at most $\varepsilon (u-n)$ (when choosing the data structure for
$S$). We say that the instance {\em covers\/} $S$ if these two
conditions hold. The number of sets that can be covered by an instance
$\mathcal{I}$ is $\binom{|U_{\mathcal{I}}|}{n}\leq \binom{\lfloor
\varepsilon (u-n)\rfloor + n}{n}$. This means that the number of
instances needed to cover all subsets of $U$ of size $n$ is
bounded from below by:
\begin{alignat*}{2}
\frac{\binom{u}{n}}{\binom{\lfloor\varepsilon (u-n)\rfloor+n}{n}} &
\geq \frac{u (u-1) \cdots (u-n+1)}{(\varepsilon (u-n) + n) (\varepsilon
(u-n) + n-1) \cdots (\varepsilon (u-n) + 1)}\\
 & > \left(\frac{u}{\varepsilon (u-n)+n}\right)^n\\
 & = \left(\frac{1}{\varepsilon}
\left(1-\frac{(1-\varepsilon)n}{\varepsilon u + (1-\varepsilon)
n}\right)\right)^n\\
 & \geq \left(\frac{1}{\varepsilon} \right)^n
\exp\left(-\frac{(1-\varepsilon)n^2}{\varepsilon u+ (1-\varepsilon)
n}\right) \enspace .
\end{alignat*}
Since each instance has a unique memory representation, this means
that the number of bits used by the data structure in the worst case
must be at least $\log_2$ of the number of instances.
$\Box$ 

%%%%%%%%%%%%%%%%%%%%%%%%%%%%%%%%%%%%%%%%%%%%%%%%%%%%%%%%%%%%%%%%%%%%%%%%%%%%%%%%%%%%%%%%%%%%%%%%


\begin{thebibliography}{99}%\setlength{\itemsep}{0pt}
%%%%%%%%%%%%%%%%%%%%%%%%%%%%%%%%%%%%%%%%%%%%%%%%%%%%%%%%%%%%%%%%%%%%%%%%%%%%%%%%%%%%%%%%%%%%%%%%%%%%
\bibitem{range-1D}
        S.\,Alstrup, G.\,S.\,Brodal, and T.\,Rauhe,
       Optimal static range reporting in one dimension,
             Proc. 33rd ACM STOC, 2001, pp.~476--482.


\bibitem{BKK92} G.\,V. Balakin, V.\,F. Kolchin, and V.\,I. Khokhlov,
Hypercycles in a random hypergraph,
\emph{Discrete Mathematics and Applications} \textbf{2} 1992: 563--570.

\bibitem{BM:02}
             A.\,Z.\,Broder and M.\,Mitzenmacher,
           Network applications of {B}loom filters: a survey, in:
   Proc. 40th Annual Allerton Conference on Communication, Control, and Computing,
 pp.~636--646, ACM Press, 2002.

\bibitem{BPZ07} F.\,C. Botelho, R.\,Pagh, and N.\,Ziviani, 
          Simple and space-efficient minimal perfect hash functions, 
          in: Proc. 10th WADS 2007,  
          Springer LNCS 4619, pp.~139--150.

\bibitem{bloom} B.\,H. Bloom,
          Space/time trade-offs in hash coding with allowable errors,
          \emph{Commun. ACM} \textbf{13}(7) 1970: 422--426.

\bibitem{Cain07} {J.\,A.\,Cain, P.\,Sanders, N.\,C.\,Wormald},
  The random graph threshold for $k$-orientiability and a fast algorithm for optimal multiple-choice allocation,
  Proc. 18th ACM-SIAM SODA, 2007, pp.~469--476.

\bibitem{Cal96} N.\,J. Calkin, 
           Dependent sets of constant weight  vectors in GF$(q)$, 
           \emph{Random Structures and Algorithms} \textbf{9} 1996: 49--54.

\bibitem{Cal97} N.\,J. Calkin, Dependent sets of constant weight 
binary vectors, \emph{Combinatorics, Probability and Computing} \textbf{6}(3) 1997: 263--271.

\bibitem{MR80h:68037} L.\,Carter, 
               R.\,W.\,Floyd,
               J.\,Gill,
               G.\,Markowsky, and
               M.\,N.\,Wegman,
              Exact and approximate membership testers,
             Proc. 10th ACM STOC, 1978, pp.~59--65.


\bibitem{CW:79} L.\,Carter and M.\,N.\,Wegman, 
           Universal classes of hash functions,
             \emph{J. Comput. Syst. Sci.} \textbf{18}(2) 1979: 143–-154.


\bibitem{CKRT04} B.\,Chazelle, J.\,Kilian, R.\,Rubinfeld, A.\,Tal, The Bloomier filter: 
an efficient data structure for static support lookup tables, 
Proc. 15th ACM-SIAM SODA, 2004, pp.~30--39.

\bibitem{Coo99} C.\,Cooper,
  Asymptotics for dependent sums of random vectors,
  \emph{Random Struct. Algorithms}  \textbf{14}(3) 1999:  267--292.


\bibitem{Coo00} C.\,Cooper,
  On the rank of random matrices,
  \emph{Random Struct. Algorithms}  \textbf{16}(2) 2001:  209--232.
  

\bibitem{CRS03} A.\,Czumaj, C.\,Riley, C.\,Scheideler, 
Perfectly Balanced Allocation, in: Proc. RANDOM-APPROX 2003, Springer LNCS 2764, pp.~240--251.


\bibitem{Die07} M.\,Dietzfelbinger,
                 Design strategies for minimal perfect hash functions,
                 in: Proc. 4th Int. Symp. on Stochastic Algorithms: Foundations and Applications (SAGA), 
                 2007, Springer LNCS 4665, pp.~2--17.
                 
\bibitem{DieWei07}        M.\,Dietzfelbinger and C.\,Weidling, 
        Balanced allocation and dictionaries 
       with tightly packed constant size bins,
       \emph{Theoret. Comput. Sci.} \textbf{380}(1--2) 2007: 47--68.
 
\bibitem{fern07} {D.\,Fernholz and
               V.\,Ramachandran},
  The $k$-orientability thresholds for $G_{n,p}$,
Proc. 18th ACM-SIAM SODA, 2007, pp.~459--468.
                              
\bibitem{FPSS05}  
       D.\,Fotakis, R.\,Pagh, P.\,Sanders, P.\,G.\,Spirakis, 
       Space efficient hash tables with worst case constant access time,
       \emph{Theory. Comput. Syst.} \textbf{38}(2) 2005: 229--248.

\bibitem{HT01}          
	T.\,Hagerup, T.\,Tholey,
	 Efficient minimal perfect hashing in nearly minimal space, in:
	Proc. 18th STACS 2001, Springer LNCS 2010, pp.~317--326.         

\bibitem{Kol94} V.\,F.\,Kolchin,
Random graphs and systems of linear equations in finite fields,
\emph{Random Structures and Algorithms} \textbf{5} 1994: 135--146.

\bibitem{KK95} V.\,F.\,Kolchin  and V.\,I.\,Khokhlov,
A threshold effect for systems of random equations 
of a special form, \emph{Discrete Mathematics and Applications} \textbf{5} 1995: 425--436.

\bibitem{p:MWHC96} 
        B.\,S.\,Majewski, N.\,C.\,Wormald,  G.\,Havas,  and Z.\,J.\,Czech, A family of perfect hashing
        methods, \emph{Computer J.} \textbf{39}(6) 1996: 547--554 

\bibitem{MU:05} M.\,Mitzenmacher and E.\,Upfal, 
       \emph{Probability and Computing},  Cambridge University Press, Cambridge,  2005.

\bibitem{Mitz:02}
	M.~Mitzenmacher,
	Compressed {B}loom filters,
	\emph{IEEE/ACM Transactions on Networking}, \textbf{10}(5):604--612 (2002).

\bibitem{range1d2} C.\,W.\,Mortensen and R.\,Pagh, and M.\,P\v{a}tra\c{s}cu,
    On dynamic range reporting in one dimension,
             Proc. 37th ACM STOC, 2005,  pp.~104--111.

\bibitem{MotRag95}
R.~Motwani and P.~Raghavan.
\newblock {\em Randomized Algorithms}.
\newblock Cambridge University Press, 1995.

\bibitem{dict-jour}
R.~Pagh.
\newblock Low redundancy in static dictionaries with constant query time.
\newblock {\em SIAM J.\ Comput.}, \textbf{31}(2):353--363, 2001.

\bibitem{cuckoo-jour}
R.\,Pagh and F.\,F.\,Rodler, Cuckoo Hashing, \emph{J. Algorithms} \textbf{51}:122--144 (2004).

\bibitem{panigrahy05} R. Panigrahy, Efficient hashing with lookups in two memory accesses, 
Proc. 16th ACM-SIAM SODA, 2005, pp.~830--839.

\bibitem{rankdict} 
              V.\,Raman and S.\,S.\,Rao,
              Static Dictionaries Supporting Rank,
            in: Proc. 10th Int. Symp. on Algorithms And Computation (ISAAC), 1999,
         Springer LNCS 1741, pp.~18--26.

\bibitem{HS94} S.\,S.\,Seiden and D.\,S.\,Hirschberg, Finding succinct ordered minimal perfect hash functions, 
\emph{Inf. Process. Lett.} \textbf{51}(6) 1994: 283--288.


\bibitem{WC:79} 
M.\,N.\,Wegman and L.\,Carter, 
           New Classes and Applications of Hash Functions, in:
             Proc. 20th IEEE FOCS, 1979, pp.~175--182.

\bibitem{ZuHeBo:DAMON:06}
M.\,Zukowski and S.\,Heman and P.\,A.\,Boncz,
Architecture-conscious hashing,
in: Proc. Int. Workshop on Data Management on New Hardware ({DaMoN}),
Chicago, 2006, Article No.~6 (8 pages).

\end{thebibliography}
\end{document}